\documentclass[twocolumn,prl,aps,floatfix]{revtex4-2}
\usepackage{amsmath,amssymb,amsfonts,amsthm}
\usepackage{graphicx}
\usepackage{tikz}
\usetikzlibrary{arrows,positioning,shapes.geometric}
\usepackage{hyperref}

\newtheorem{theorem}{Theorem}

\newtheorem{proposition}{Proposition}
\newtheorem{corollary}{Corollary}
\newtheorem{remark}{Remark}
\newtheorem{definition}{Definition}
\begin{document}

\title{Global Non-Identifiability of Fubini–Study Geometry from Complete One-Period Endpoint Data}
\author{Rashid Ahmad}
\affiliation{Department of Physics, College of Integrative Studies, Abdullah Al Salem University, Kuwait}
\date{\today}

\begin{abstract}
We establish a global, worst-case non-identifiability theorem for periodically driven finite-dimensional quantum systems. On the unrestricted smooth periodic Hamiltonian class, we show that the period-averaged Fubini--Study metric component associated with a fixed-initial-state trajectory cannot, in general, be reconstructed from exact one-period propagators indexed by every possible starting time and by the external parameter. Thus, even complete starting-time-resolved endpoint information is insufficient to determine this intra-period geometric quantity. The obstruction is characterized exactly. The starting-time-indexed endpoint data determine the conjugation-orbit path of the monodromy but not its particular unitary representative. On each fixed-monodromy slice, the observational fibres are precisely the right orbits generated by smooth parameter-dependent based loops taking values in the pointwise centralizer of the monodromy. We show that the Fubini--Study functional is not invariant under this fibre action and therefore does not factor through the endpoint observation map. An explicit real-analytic two-level witness demonstrates the obstruction within a commuting, one-generator Hamiltonian family, so neither non-Abelian time ordering nor quasienergy-branch ambiguity is required. An entire continuous family of Hamiltonians produces identical starting-time-indexed one-period endpoint data while yielding different, and in fact unboundedly separated, period-averaged Fubini--Study geometry. The non-identifiability persists under a uniform bound on the parameter derivative of the Hamiltonian. The result is a deterministic global identifiability statement rather than a claim of generic or experimental impossibility. It identifies a precise information gap between complete starting-time-indexed one-period endpoint data and the full intra-period propagator, and clarifies that the missing information is the unitary lift of the observed monodromy orbit rather than ordinary scalar phase freedom.
\end{abstract}

\maketitle

\section{Introduction}

Periodically driven quantum systems provide a natural setting in which stroboscopic dynamics and intra-period evolution contain distinct physical information. The foundations of Floquet theory and its quantum formulation were established by Floquet, Shirley, and Sambe \cite{Floquet1883,Shirley1965,Sambe1973}. Periodic driving has subsequently become a powerful tool for engineering effective Hamiltonians, synthetic gauge fields, and nonequilibrium topological phases \cite{Kitagawa2010,LindnerRefaelGalitski2011,GoldmanDalibard2014,OkaKitamura2019}. Modern Floquet theory emphasizes that the effective Hamiltonian and the micromotion constitute distinct parts of the dynamics, with the latter depending explicitly on the phase within the driving cycle \cite{BukovDAlessioPolkovnikov2015,EckardtAnisimovas2015,Eckardt2017,Holthaus2016}. In particular, the full evolution over a cycle can contain information that is not encoded in the one-period Floquet operator or quasienergy spectrum alone \cite{Rudner2013,NathanRudner2015}.

This distinction is especially relevant for geometric quantities. The Riemannian structure of quantum-state manifolds was formulated by Provost and Vallée \cite{ProvostVallee1980}, while geometric quantum evolution in projective Hilbert space was developed by Aharonov and Anandan \cite{AharonovAnandan1987}. Berry's geometric phase established a complementary geometric description of parameter-dependent quantum evolution \cite{Berry1984,Simon1983}. The quantum geometric tensor subsequently became an important framework for quantum phase transitions, adiabatic response, parameter sensitivity, and quantum metrology \cite{BraunsteinCaves1994,ZanardiPaunkovic2006,ZanardiGiordaCozzini2007,Kolodrubetz2017}. Its geometric content has also been directly probed experimentally in photonic, superconducting, and spin-qubit systems \cite{Bleu2018,Tan2019,Yu2019,Gianfrate2020}. These developments emphasize that quantum geometry is a property of the quantum-state trajectory and can therefore depend on dynamical information beyond spectral or endpoint data.

Here we consider the corresponding inverse problem for a finite-dimensional periodically driven quantum system. The target is the period average of a parameter-space Fubini--Study metric component generated by the intra-period trajectory of a fixed initial pure state. We deliberately consider an observation model stronger than the conventional Floquet monodromy: the observer is given the exact one-period propagator for every starting time within the driving period, together with its complete dependence on the external parameter. The question is whether this complete starting-time-indexed endpoint information is sufficient to determine the geometric target.

We show that, on the unrestricted smooth finite-dimensional periodic Hamiltonian class, it is not. Specifically, there exist distinct Hamiltonian families that produce exactly the same one-period propagator for every starting time and every parameter value, while producing different period-averaged Fubini--Study metric components for the same prescribed initial state. Thus the geometric functional does not factor through the starting-time-indexed endpoint observation map. This is an exact deterministic non-identifiability result and does not rely on experimental noise, finite sampling, or numerical reconstruction error.

The obstruction has a precise global structure. The starting-time-indexed endpoint data determine the path obtained by conjugating the monodromy with the intra-period propagator, but they do not determine a particular unitary representative, or lift, of that path. On a fixed-monodromy slice, we show that the observational fibres are exactly generated by smooth parameter-dependent based loops taking values in the pointwise centralizer of the monodromy. This gives an exact characterization of observational equivalence and reduces identifiability of any target functional to invariance under the corresponding fibre action.

This ambiguity is distinct from ordinary scalar phase freedom. Scalar phases are already invisible to the Fubini--Study metric, whereas non-scalar elements of the monodromy centralizer can alter the physical projective trajectory while leaving every observed one-period endpoint propagator unchanged. The missing information is therefore not merely a quantum phase, but information about the intra-period unitary lift of the observed conjugation path.

We establish the failure of identifiability through an explicit real-analytic two-level construction. The witness lies in a one-generator commuting algebra, so Hamiltonians at different times commute and no non-Abelian time ordering is involved. Nevertheless, a continuous family of Hamiltonians has exactly identical starting-time-indexed endpoint data, while the period-averaged Fubini--Study metric varies nontrivially and becomes unbounded along the observational fibre. The construction therefore separates the present obstruction from both quasienergy-branch ambiguity and genuinely noncommutative Floquet dynamics.

The non-identifiability also survives quantitative restrictions on the Hamiltonian family. A uniform bound on the parameter derivative of the Hamiltonian does not restore identifiability, although it does provide an a priori bound on the geometric target. Conversely, within the identity observational fibre the target can attain arbitrarily large values at a prescribed parameter point. These results show that the obstruction concerns the information content of the observation map itself rather than merely the size of a particular counterexample.

The result is deliberately a global, worst-case statement. We do not claim that non-identifiability occurs on every observational fibre, nor do we establish generic non-identifiability in a specified topology or measure. Likewise, exact non-identifiability should not be confused with an instability theorem for noisy inverse reconstruction. Additional measurements that resolve the intra-period propagator can restore the target. The result instead identifies a precise limitation of complete one-period endpoint observations, even when those observations are available noiselessly and for every starting phase.

The resulting observation hierarchy is therefore strict: fixed-origin monodromy contains less information than the complete starting-time-indexed one-period endpoint data, while the latter contain less information than the full parameter- and time-resolved propagator. The geometric functional considered here does not factor through either of the first two levels, whereas it is directly determined by the full propagator. The theorem thus identifies a sharp information gap between complete one-period endpoint data and genuine micromotion-resolved dynamics.

The remainder of the paper establishes the Hamiltonian--propagator correspondence, characterizes the exact centralizer fibres of the starting-time-indexed observation map, and formulates the associated factorization criterion for identifiability. We then prove the global non-identifiability theorem using the analytic two-level witness and establish the bounded-derivative and identity-fibre extensions. Finally, we discuss the physical distinction between endpoint observational equivalence, micromotion information, and genuine gauge freedom.

\section{Setup and Definitions}

Let $\mathcal H\simeq\mathbb C^d$, with $d\ge2$, and let $I\subset\mathbb R$ be a nonempty open interval. We fix a period $T>0$, set $\hbar=1$, and take $\theta$ to be dimensionless. Starting-time indexing refers to the initial time $\tau$ of a one-period propagator, defined modulo the driving period,
\begin{align}
\mathbb T_T:=\mathbb R/T\mathbb Z .
\end{align}
Thus $[\tau]\in\mathbb T_T$ specifies the starting phase of the drive. For a smooth $T$-periodic Hamiltonian,
\begin{align}
\mathcal D_T^\infty
:=
\Bigl\{
H\in C^\infty
\bigl(I\times\mathbb R;\operatorname{Herm}(\mathbb C^d)\bigr):
\nonumber\\
H(\theta,t+T)=H(\theta,t)
\Bigr\},
\end{align}
let $U_\theta(t_2,t_1)$ denote the corresponding propagator and write
\begin{align}
U_\theta(t):=U_\theta(t,0),
\qquad
M_\theta:=U_\theta(T,0).
\end{align}
Periodicity of the Hamiltonian gives
\begin{align}
U_\theta(t+T)
=
U_\theta(t)M_\theta .
\label{eq:floquet_property}
\end{align}

We consider three levels of dynamical information. The fixed-origin monodromy is
\begin{align}
E_0:\mathcal D_T^\infty
&\longrightarrow
\mathcal Y_0:=C^\infty(I,U(d)),
\nonumber\\
E_0[H](\theta)
&:=U_\theta(T,0)=M_\theta .
\label{eq:E0}
\end{align}
A strictly richer observation is the complete starting-time-indexed family of one-period propagators,
\begin{align}
E_{\rm ph}:\mathcal D_T^\infty
&\longrightarrow
\mathcal Y_{\rm ph}
:=C^\infty(I\times\mathbb T_T,U(d)),
\nonumber\\
E_{\rm ph}[H](\theta,[\tau])
&:=U_\theta(\tau+T,\tau).
\label{eq:Eph}
\end{align}
The periodicity of the drive implies
\begin{align}
U_\theta(\tau+2T,\tau+T)
=
U_\theta(\tau+T,\tau),
\end{align}
so $E_{\rm ph}$ is well defined on $\mathbb T_T$. This observation specifies the complete one-period endpoint data for every starting phase, but contains no subperiod propagators.

For comparison, the full intraperiod propagator is
\begin{align}
E_{\rm full}:\mathcal D_T^\infty
&\longrightarrow
\mathcal Y_{\rm full}
:=C^\infty(I\times[0,T],U(d)),
\nonumber\\
E_{\rm full}[H](\theta,t)
&:=U_\theta(t,0).
\label{eq:Efull}
\end{align}
Since
\begin{align}
U_\theta(t_2,t_1)
=
U_\theta(t_2,0)
U_\theta(t_1,0)^\dagger ,
\end{align}
$E_{\rm full}$ determines the complete two-time evolution within a period.

The inverse problem considered here is formulated at the level of exact unitary data. In particular, $E_{\rm ph}$ is assumed known as an exact smooth map, with no statistical or experimental uncertainty. The resulting identifiability statement is therefore a deterministic one; questions of reconstruction stability under noise are separate. The unitary formulation is also deliberately stronger than observations restricted to projective unitaries or quantum channels.

For a fixed pure initial state
\begin{align}
\rho_0=|\psi_0\rangle\langle\psi_0|,
\end{align}
define the parameter generator in the instantaneous body frame by
\begin{align}
K_\theta(t)
:=
iU_\theta^\dagger(t)\partial_\theta U_\theta(t).
\label{eq:K}
\end{align}
The target functional is the period average of the $\theta\theta$ component of the Fubini--Study metric pulled back along the fixed-state trajectory,
\begin{align}
G_{\rho_0}[H](\theta)
:=
\frac{1}{T}
\int_0^T
\operatorname{Var}_{\rho_0}
\!\left(K_\theta(t)\right)\,dt ,
\label{eq:G}
\end{align}
where
\begin{align}
\operatorname{Var}_{\rho_0}(X)
:=
\langle\psi_0|X^2|\psi_0\rangle
-
\langle\psi_0|X|\psi_0\rangle^2 .
\end{align}
Equivalently, the parameter-dependent ray
\begin{align}
\iota_t:I\longrightarrow\mathbb{CP}^{d-1},
\qquad
\iota_t(\theta)
:=[U_\theta(t)|\psi_0\rangle]
\end{align}
induces the metric component
\begin{align}
G_{\rho_0}[H](\theta)
=
\frac{1}{T}
\int_0^T
(\iota_t^*g_{\rm FS})_{\theta\theta}\,dt .
\label{eq:G_FS}
\end{align}
With the convention
\begin{align}
g_{\rm FS}
(\partial_\theta\psi,\partial_\theta\psi)
=
\langle\partial_\theta\psi|\partial_\theta\psi\rangle
-
|\langle\psi|\partial_\theta\psi\rangle|^2 ,
\end{align}
the pure-state quantum Fisher information satisfies $F_Q=4g_{\rm FS}$. Hence
\begin{align}
G_{\rho_0}[H](\theta)
=
\frac{1}{4T}
\int_0^T
F_Q[\rho_\theta(t)]\,dt ,
\quad
\rho_\theta(t)
=
U_\theta(t)\rho_0U_\theta^\dagger(t).
\label{eq:QFI}
\end{align}
Thus the target is the time average of the instantaneous pure-state quantum Fisher information and should not be confused with the quantum Fisher information of a time-averaged state.

\begin{remark}[Coordinate dependence]
$G_{\rho_0}$ is a coordinate component of the pullback metric rather than a scalar quantity. Under a smooth reparametrization $\tilde\theta=\tilde\theta(\theta)$,
\begin{align}
G_{\tilde\theta\tilde\theta}
=
\left(
\frac{d\theta}{d\tilde\theta}
\right)^2
G_{\theta\theta}.
\label{eq:coordinate_transform}
\end{align}
All identifiability statements below therefore refer to the prescribed parameter coordinate $\theta$.
\end{remark}

\begin{remark}[Reference-time convention]
The observation $E_{\rm ph}$ contains one-period propagators for every starting phase, whereas $G_{\rho_0}$ is defined from the trajectory generated from the fixed reference time $t=0$. A change of reference time would in general define a different functional. The inclusion of all starting phases in $E_{\rm ph}$ therefore does not alter the definition of the target.
\end{remark}

\begin{remark}[Projective gauge and observational fibres]
The scalar $U(1)$ phase is a projective gauge and does not affect the physical ray or the Fubini--Study metric. The ambiguity relevant here is instead an observational equivalence associated with the centralizer of the monodromy. In particular, a transformation
\begin{align}
U_\theta(t)
\longmapsto
U_\theta(t)V_\theta(t)
\end{align}
by a suitable centralizer-valued based loop can leave $E_{\rm ph}$ unchanged while generally altering the ray
\begin{align}
[U_\theta(t)|\psi_0\rangle].
\end{align}
Consequently, equality of the complete starting-time-indexed endpoint data does not imply equality of the underlying intraperiod trajectory. This distinction between projective gauge freedom and endpoint-induced observational equivalence is the key structural feature underlying the non-identifiability result.
\end{remark}

\section{Hamiltonian--Propagator Bijection}

It is convenient to formulate the inverse problem directly at the level of propagators. Define
\begin{multline}
\mathcal P_T^\infty
:=
\Bigl\{
U\in C^\infty(I\times\mathbb R,U(d)):
\\
U(\theta,0)=\mathbb I,\;
U(\theta,t+T)=U(\theta,t)U(\theta,T)
\Bigr\}.
\end{multline}

\begin{proposition}[Hamiltonian--propagator correspondence]
The map
\begin{align}
\Phi:\mathcal D_T^\infty&\longrightarrow\mathcal P_T^\infty,
&
\Phi(H)(\theta,t)&=U_\theta(t,0),
\end{align}
is a bijection. Its inverse is
\begin{align}
\Phi^{-1}(U)(\theta,t)
&=
i\,\partial_tU_\theta(t)U_\theta^\dagger(t)
=
-i\,U_\theta(t)\partial_tU_\theta^\dagger(t).
\label{eq:H_from_U}
\end{align}
\end{proposition}

\begin{proof}
For $H\in\mathcal D_T^\infty$, existence and uniqueness of the Schrödinger equation give a unique smooth propagator $U$, with $U(\theta,0)=\mathbb I$. Periodicity of $H$ implies the Floquet relation in Eq.~\eqref{eq:floquet_property}, so $\Phi(H)\in\mathcal P_T^\infty$.

Conversely, let $U\in\mathcal P_T^\infty$ and define
\begin{align}
H_U(\theta,t)
:=
i\,\partial_tU_\theta(t)U_\theta^\dagger(t).
\end{align}
Unitarity gives
\begin{align}
\partial_tU_\theta U_\theta^\dagger
+
U_\theta\partial_tU_\theta^\dagger
&=0,
\\
H_U^\dagger
&=
-iU_\theta\partial_tU_\theta^\dagger
=
i\,\partial_tU_\theta U_\theta^\dagger
=
H_U,
\end{align}
so $H_U$ is Hermitian. Moreover, using
$U_\theta(t+T)=U_\theta(t)U_\theta(T)$,
\begin{align}
\partial_tU_\theta(t+T)
&=
\partial_tU_\theta(t)U_\theta(T),
\\
H_U(\theta,t+T)
&=
i\,\partial_tU_\theta(t)
U_\theta(T)U_\theta^\dagger(T)
U_\theta^\dagger(t)
\nonumber\\
&=
i\,\partial_tU_\theta(t)U_\theta^\dagger(t)
=
H_U(\theta,t).
\end{align}
Hence $H_U\in\mathcal D_T^\infty$. The two constructions are inverse to one another, proving the bijection.
\end{proof}

The normalization $U(\theta,0)=\mathbb I$ is essential: without it, the same propagator trajectory would admit an arbitrary $\theta$-dependent left unitary factor and would not determine a unique Hamiltonian in the stated representation.

For a fixed smooth monodromy $M\in C^\infty(I,U(d))$, define the fixed-monodromy slice
\begin{align}
\mathcal P_M^\infty
:=
\Bigl\{
U\in\mathcal P_T^\infty:
U_\theta(T)=M_\theta
\Bigr\}.
\end{align}
The corresponding pointwise unitary and Lie-algebra centralizers are
\begin{align}
Z(M_\theta)
&:=
\{X\in U(d):XM_\theta=M_\theta X\},
\\
\mathfrak z(M_\theta)
&:=
\{X\in\mathfrak u(d):[X,M_\theta]=0\}.
\end{align}
At the propagator level, the starting-time-indexed observation is
\begin{align}
E_{\rm ph}:\mathcal P_T^\infty
&\longrightarrow
\mathcal Y_{\rm ph},
\nonumber\\
E_{\rm ph}[U](\theta,[\tau])
&=
U_\theta(\tau+T,\tau).
\label{eq:Eph_prop}
\end{align}
The Hamiltonian and propagator formulations are equivalent by the preceding proposition.

\section{Observation Geometry and Exact Fibre Characterization}

The structure of the starting-time-indexed endpoint data follows directly from the Floquet relation.

\begin{proposition}[Starting-time-indexed observation]
For every $U\in\mathcal P_T^\infty$,
\begin{align}
E_{\rm ph}[U](\theta,[\tau])
&=
U_\theta(\tau)M_\theta U_\theta^\dagger(\tau),
\label{eq:orbit_observation}
\end{align}
where $M_\theta=U_\theta(T)$.
\end{proposition}

\begin{proof}
Using the composition law and the Floquet relation,
\begin{align}
U_\theta(\tau+T,\tau)
&=
U_\theta(\tau+T,0)
U_\theta^\dagger(\tau,0)
\nonumber\\
&=
U_\theta(\tau)M_\theta
U_\theta^\dagger(\tau).
\end{align}
\end{proof}

Define the observed conjugation-orbit path
\begin{align}
\Gamma_U(\theta,[t])
:=
\operatorname{Ad}_{U_\theta(t)}M_\theta
=
U_\theta(t)M_\theta U_\theta^\dagger(t).
\end{align}
Equation~\eqref{eq:orbit_observation} shows that
\begin{align}
E_{\rm ph}[U]=\Gamma_U.
\end{align}
Thus, for each $\theta$, the complete starting-time-indexed endpoint data determine the path of $M_\theta$ through its conjugacy orbit
\begin{align}
\mathcal O_{M_\theta}
:=
\{VM_\theta V^\dagger:V\in U(d)\},
\end{align}
but do not, in general, determine the unitary $U_\theta(t)$ that generates this path.

The resulting ambiguity is characterized by the elementary identity
\begin{align}
\operatorname{Ad}_{U_1(t)}M
=
\operatorname{Ad}_{U_2(t)}M
\quad\Longleftrightarrow\quad
[U_2^\dagger(t)U_1(t),M]=0.
\label{eq:orbit_fibre_identity}
\end{align}
Hence the relevant fibres are generated by right multiplication with paths taking values in the centralizer of the monodromy.

An infinitesimal form of the same observation follows by differentiating the orbit path:
\begin{align}
\dot\Gamma(t)
&=
\dot U(t)MU^\dagger(t)
+
UM\dot U^\dagger(t)
\nonumber\\
&=
-i[H(t),\Gamma(t)].
\label{eq:Gamma_dot}
\end{align}
Consequently, the observed tangent vector determines only the commutator component of the Hamiltonian. Its kernel is the Lie-algebra centralizer
\begin{align}
\ker\operatorname{ad}_{\Gamma(t)}
=
\mathfrak z(\Gamma(t)).
\end{align}
Since
\begin{align}
\mathfrak z(\Gamma(t))
=
U(t)\mathfrak z(M)U^\dagger(t),
\end{align}
the infinitesimal ambiguity is precisely the tangent-space manifestation of the global centralizer fibre. The infinitesimal relation is useful for interpretation but is not required for the global fibre characterization.

\begin{theorem}[Fixed-monodromy fibre characterization]
Fix $M\in C^\infty(I,U(d))$. On $\mathcal P_M^\infty$, define
\begin{align}
\mathcal G_M^0
:=
\Bigl\{
V\in C^\infty(I\times\mathbb R,U(d)):
&\;V_\theta(0)=\mathbb I,
\nonumber\\
&\;V_\theta(t+T)=V_\theta(t),
\nonumber\\
&\;[V_\theta(t),M_\theta]=0
\quad\forall(\theta,t)
\Bigr\}.
\label{eq:G_M}
\end{align}
Then, for $U_1,U_2\in\mathcal P_M^\infty$,
\begin{align}
E_{\rm ph}[U_1]=E_{\rm ph}[U_2]
\quad\Longleftrightarrow\quad
U_1=U_2V
\quad\text{for some }V\in\mathcal G_M^0.
\label{eq:fibre_characterization}
\end{align}
Thus every fibre of $E_{\rm ph}$ on $\mathcal P_M^\infty$ is a right $\mathcal G_M^0$-orbit.
\end{theorem}

\begin{proof}
Suppose first that
$E_{\rm ph}[U_1]=E_{\rm ph}[U_2]$. By Eq.~\eqref{eq:orbit_observation},
\begin{align}
U_1(t)MU_1^\dagger(t)
=
U_2(t)MU_2^\dagger(t).
\end{align}
Define
\begin{align}
V(t):=U_2^\dagger(t)U_1(t).
\end{align}
Then
\begin{align}
V(t)M
=
MV(t),
\end{align}
so $V(t)$ takes values in $Z(M)$. Since both propagators are normalized at $t=0$,
\begin{align}
V(0)=\mathbb I.
\end{align}
Furthermore,
\begin{align}
V(t+T)
&=
U_2^\dagger(t+T)U_1(t+T)
\nonumber\\
&=
(U_2(t)M)^\dagger(U_1(t)M)
\nonumber\\
&=
M^\dagger V(t)M
=
V(t),
\end{align}
where the last equality follows from $[V(t),M]=0$. Hence $V\in\mathcal G_M^0$ and $U_1=U_2V$.

Conversely, suppose $U_1=U_2V$ with $V\in\mathcal G_M^0$. Then
\begin{align}
U_1(t)MU_1^\dagger(t)
&=
U_2(t)V(t)MV^\dagger(t)U_2^\dagger(t)
\nonumber\\
&=
U_2(t)MU_2^\dagger(t),
\end{align}
so $E_{\rm ph}[U_1]=E_{\rm ph}[U_2]$. Moreover,
\begin{align}
U_1(t+T)
&=
U_2(t)MV(t)
\nonumber\\
&=
U_2(t)V(t)M
=
U_1(t)M,
\end{align}
and therefore $U_1\in\mathcal P_M^\infty$.
\end{proof}

\begin{corollary}[Global fibre characterization]
For $U\in\mathcal P_T^\infty$, let
\begin{align}
M_U(\theta):=U_\theta(T).
\end{align}
Since
\begin{align}
E_{\rm ph}[U](\theta,[0])=M_U(\theta),
\end{align}
the monodromy is constant on every $E_{\rm ph}$ fibre. Consequently,
\begin{align}
E_{\rm ph}^{-1}(E_{\rm ph}[U])
\cap
\mathcal P_{M_U}^\infty
=
U\,\mathcal G_{M_U}^0 .
\label{eq:global_fibre}
\end{align}
\end{corollary}

The right action of $\mathcal G_M^0$ is free and transitive on each fibre. Indeed, if $UV=U$, then $V=\mathbb I$, while for any two elements $U_1,U_2$ of the same fibre, the unique element relating them is
\begin{align}
V=U_2^\dagger U_1.
\end{align}
This is a set-theoretic orbit statement; no manifold or quotient structure on $\mathcal P_M^\infty/\mathcal G_M^0$ is required.

It follows immediately that a functional $F$ on a fixed-monodromy slice is identifiable from $E_{\rm ph}$ if and only if it is invariant under this right action,
\begin{align}
F[UV]=F[U]
\qquad
\forall\,V\in\mathcal G_M^0.
\label{eq:fibre_invariance}
\end{align}
Equivalently, $F$ must descend to the set-theoretic quotient
$\mathcal P_M^\infty/\mathcal G_M^0$.

\subsection{Centralizer-Valued Based Loops}

The group $\mathcal G_M^0$ consists of smooth based loops whose values lie pointwise in the centralizer family
\begin{align}
Z(M_\theta)
=
\{X\in U(d):XM_\theta=M_\theta X\}.
\end{align}
When $M_\theta$ has simple spectrum, this centralizer is isomorphic to the maximal torus $U(1)^d$; at spectral degeneracies it is enlarged. The argument does not require a global trivialization of the resulting family of centralizers.

The induced transformation at the Hamiltonian level makes the physical content of the fibre explicit. For
\begin{align}
U'(t)=U(t)V(t),
\qquad
V\in\mathcal G_M^0,
\end{align}
the corresponding Hamiltonian is
\begin{align}
H'
&=
i\dot U' U'^\dagger
\nonumber\\
&=
H+iU\dot V V^\dagger U^\dagger .
\label{eq:H_transformation}
\end{align}
Because $V$ is periodic and satisfies $V(0)=\mathbb I$, the transformed propagator has the same monodromy and belongs to the same Hamiltonian class. In particular,
\begin{align}
E_0[H']
=
E_0[H],
\qquad
E_{\rm ph}[H']
=
E_{\rm ph}[H].
\end{align}

The corresponding transformation of the body-frame parameter generator is
\begin{align}
K'
&=
i(UV)^\dagger
\partial_\theta(UV)
\nonumber\\
&=
V^\dagger K V
+
iV^\dagger\partial_\theta V .
\label{eq:K_transformation}
\end{align}
The second term is Hermitian by unitarity of $V$. Thus the fibre action modifies the generator both by conjugation and by an inhomogeneous parameter-space connection term. Since the variance in the target is evaluated in the fixed state $\rho_0$, neither contribution is invariant in general.

\begin{remark}[Observational equivalence versus gauge equivalence]
The centralizer action is an observational equivalence induced by the endpoint map, rather than a gauge symmetry of the underlying dynamics. In general, $U\mapsto UV$ changes the physical ray $[U(t)|\psi_0\rangle]$ and the associated Hamiltonian $H\mapsto H'$. For each fixed $\theta$, the observation determines a path in the conjugacy orbit
\begin{align}
\mathcal O_{M_\theta}
\simeq
U(d)/Z(M_\theta),
\end{align}
whereas the propagator specifies a lift of that path to $U(d)$. The endpoint observation fixes the orbit path but not, in general, this lift. This distinction is the structural origin of the non-identifiability established below.
\end{remark}

\section{Universal Fibre Identifiability Criterion}

We now state the general criterion that reduces identifiability from the starting-time-indexed endpoint data to invariance under the corresponding observational fibres.

\begin{definition}[Identifiability from endpoint data]
A functional
\begin{align}
F:\mathcal D_T^\infty\longrightarrow\mathcal Y
\end{align}
is said to be \emph{identifiable from} $E_{\mathrm{ph}}$ if there exists a deterministic map
\begin{align}
R:\mathcal Y_{\mathrm{ph}}\longrightarrow\mathcal Y
\end{align}
such that
\begin{align}
F[H]=R\bigl(E_{\mathrm{ph}}[H]\bigr)
\qquad
\text{for all }H\in\mathcal D_T^\infty .
\end{align}
Equivalently, $F$ must take the same value on any two Hamiltonians that produce identical endpoint data.
\end{definition}

\begin{proposition}[Fibrewise identifiability criterion]
Let $M\in C^\infty(I,U(d))$ and consider the fixed-monodromy slice
$\mathcal P_M^\infty$. A functional
\begin{align}
F:\mathcal P_M^\infty\longrightarrow\mathcal Y
\end{align}
is determined by the starting-time-indexed endpoint observation $E_{\mathrm{ph}}$ on this slice if and only if
\begin{align}
F[UV]=F[U]
\qquad
\text{for all }U\in\mathcal P_M^\infty,\quad
V\in\mathcal G_M^0 .
\label{eq:fibrewise_identifiability}
\end{align}
Thus identifiability is equivalent to invariance of the target functional under the complete fibre action characterized above.
\end{proposition}

\begin{proof}
By the fixed-monodromy fibre characterization, the fibre through any
$U\in\mathcal P_M^\infty$ is precisely the right orbit
$U\mathcal G_M^0$. Hence $E_{\mathrm{ph}}[U_1]=E_{\mathrm{ph}}[U_2]$ if and only if
$U_2=U_1V$ for some $V\in\mathcal G_M^0$. Therefore, $F$ is constant on the fibres of $E_{\mathrm{ph}}$ if and only if Eq.~\eqref{eq:fibrewise_identifiability} holds. In that case, $F$ descends uniquely to the corresponding set of endpoint-data equivalence classes, and consequently can be obtained by deterministic post-processing of $E_{\mathrm{ph}}$. Conversely, if Eq.~\eqref{eq:fibrewise_identifiability} fails for some $U$ and $V$, then $U$ and $UV$ produce identical endpoint data but different values of $F$, so no deterministic reconstruction from $E_{\mathrm{ph}}$ is possible.
\end{proof}

This criterion is the basic identifiability test used below: to establish non-identifiability, it suffices to exhibit a single admissible fibre transformation that leaves the complete endpoint observation unchanged while altering the target functional.

\section{Main Theorem and Analytic Two-Level Witness}

We next give an explicit analytic witness for the non-identifiability established above. The construction is deliberately elementary: it is confined to a two-dimensional invariant subspace and involves a Hamiltonian that is Abelian at all times. Thus, neither noncommutativity nor Floquet-branch ambiguities are required for the obstruction.

Let $f(t)=\sin(2\pi t/T)$ and let
$\rho_0=|\psi_0\rangle\langle\psi_0|$ be any prescribed pure state. Choose a normalized state $|\phi\rangle$ orthogonal to $|\psi_0\rangle$ and define
\begin{align}
A&=|\psi_0\rangle\langle\phi|
+|\phi\rangle\langle\psi_0|.
\end{align}
Then $A=A^\dagger$, $\|A\|_{\mathrm{op}}=1$, and
$\operatorname{Var}_{\rho_0}(A)=1$. For $d>2$, $A$ is understood to vanish on the orthogonal complement of
$\operatorname{span}\{|\psi_0\rangle,|\phi\rangle\}$.

For $\alpha\in\mathbb R$, consider the propagator
\begin{align}
U_{\theta,\alpha}(t)
&=
\exp\left\{
-i\theta\bigl[t+\alpha T f(t)\bigr]A
\right\}.
\label{eq:witness_propagator}
\end{align}
The corresponding Hamiltonian is
\begin{align}
H_{\theta,\alpha}(t)
&=
i\,\partial_tU_{\theta,\alpha}(t)
U_{\theta,\alpha}^\dagger(t)
\nonumber\\
&=
\theta
\left[
1+2\pi\alpha\cos\left(\frac{2\pi t}{T}\right)
\right]A .
\label{eq:witness_hamiltonian}
\end{align}
These Hamiltonians are real analytic in $(\theta,t)$ and $T$-periodic in time. Moreover, they take values in the one-dimensional Abelian algebra generated by $A$, and therefore
\begin{align}
[H_{\theta,\alpha}(t),H_{\theta,\alpha}(s)]=0
\end{align}
for all $s,t$. The construction consequently involves no nontrivial time ordering.

The one-period propagator is independent of $\alpha$:
\begin{align}
M_\theta
&=
U_{\theta,\alpha}(T)
=
e^{-i\theta T A}.
\label{eq:witness_monodromy}
\end{align}
In particular,
\begin{align}
\partial_\theta M_\theta
=
-iTAe^{-i\theta TA},
\end{align}
which is nonzero for every $\theta\in I$. Thus the parameter dependence of the monodromy is nontrivial on every nonempty open subinterval of $I$.

More strongly, the complete starting-time-indexed endpoint data are independent of $\alpha$. Indeed, using the periodicity of $f$,
\begin{align}
U_{\theta,\alpha}(\tau+T,\tau)
&=
U_{\theta,\alpha}(\tau+T)
U_{\theta,\alpha}^\dagger(\tau)
\nonumber\\
&=
e^{-i\theta T A}
=
M_\theta ,
\label{eq:witness_endpoint}
\end{align}
for every $\tau\in\mathbb R$. Hence
\begin{align}
E_{\rm ph}[H_\alpha]
=
E_{\rm ph}[H_0]
\end{align}
as exact smooth maps on $I\times\mathbb T_T$.

The corresponding body-frame generator is
\begin{align}
K_{\theta,\alpha}(t)
&=
iU_{\theta,\alpha}^\dagger(t)
\partial_\theta U_{\theta,\alpha}(t)
\nonumber\\
&=
\bigl[t+\alpha T f(t)\bigr]A .
\label{eq:witness_generator}
\end{align}
Consequently, the period-averaged Fubini--Study component is
\begin{align}
G_{\rho_0}[H_\alpha](\theta)
&=
T^2\operatorname{Var}_{\rho_0}(A)
\left(
\frac{1}{3}-\frac{\alpha}{\pi}
+\frac{\alpha^2}{2}
\right).
\label{eq:witness_metric}
\end{align}
It follows that
\begin{align}
G_{\rho_0}[H_\alpha](\theta)
-
G_{\rho_0}[H_0](\theta)
&=
T^2\operatorname{Var}_{\rho_0}(A)
\left(
\frac{\alpha^2}{2}-\frac{\alpha}{\pi}
\right).
\label{eq:witness_difference}
\end{align}
For the chosen state, $\operatorname{Var}_{\rho_0}(A)=1$, so the difference is nonzero except at $\alpha=0$ and $\alpha=2/\pi$, and grows quadratically with $|\alpha|$.

The same construction has a direct interpretation in terms of the fibre characterized above. Writing
\begin{align}
U_{\theta,\alpha}(t)
&=
U_{\theta,0}(t)V_{\theta,\alpha}(t),
\nonumber\\
V_{\theta,\alpha}(t)
&=
e^{-i\alpha\theta T f(t)A},
\end{align}
we have
\begin{align}
V_{\theta,\alpha}(0)
&=I,
&
V_{\theta,\alpha}(t+T)
&=V_{\theta,\alpha}(t),
&
[V_{\theta,\alpha}(t),M_\theta]
&=0.
\end{align}
Thus $V_{\theta,\alpha}\in\mathcal G_M^0$, and the two propagators lie in the same fibre of $E_{\rm ph}$ while producing different values of the target functional.

\begin{theorem}[Global non-identifiability]
For every $d\ge2$, $T>0$, every nonempty open interval $I\subset\mathbb R$, and every prescribed pure state $\rho_0$, there exist real-analytic, $T$-periodic Hamiltonian families
$H_0,H_\alpha\in\mathcal D_T^\infty$, jointly real analytic in $(\theta,t)$, such that
\begin{align}
[H_{\theta,\alpha}(t),H_{\theta,\alpha}(s)]=0
\end{align}
for all $\theta\in I$ and all $s,t\in\mathbb R$, while
\begin{align}
E_{\rm ph}[H_\alpha]
&=
E_{\rm ph}[H_0]
\end{align}
as exact $C^\infty(I\times\mathbb T_T,U(d))$-valued data, but
\begin{align}
G_{\rho_0}[H_\alpha](\theta)
-
G_{\rho_0}[H_0](\theta)
&=
T^2\operatorname{Var}_{\rho_0}(A)
\left(
\frac{\alpha^2}{2}-\frac{\alpha}{\pi}
\right).
\end{align}
Hence $G_{\rho_0}$ is not identifiable from $E_{\rm ph}$ on the unrestricted smooth periodic Hamiltonian class. In fact, the ambiguity can be made arbitrarily large by increasing $|\alpha|$.
\end{theorem}

\begin{proof}
The explicit family above satisfies all required regularity and periodicity conditions. Equation~\eqref{eq:witness_endpoint} shows that every member of the family has identical starting-time-indexed endpoint data, whereas Eq.~\eqref{eq:witness_difference} shows that the target functional is not constant on this fibre. The fibrewise identifiability criterion therefore rules out any deterministic reconstruction of $G_{\rho_0}$ from $E_{\rm ph}$.
\end{proof}

The witness is supported on a two-dimensional invariant subspace and belongs to the single-generator commuting algebra $\mathbb R A$. Thus the non-identifiability does not rely on noncommuting dynamics, quasienergy-branch choices, or complicated Floquet structure. The parameter enters the Hamiltonian linearly, while the auxiliary periodic deformation changes the intra-period parameter response without changing the complete family of one-period endpoint propagators.

The quantifiers are important: the result has the form
$\forall\rho_0,\exists(H_0,H_\alpha)$. The witness pair may therefore depend on the prescribed initial state. The theorem does not claim the existence of a single pair of Hamiltonian families that distinguishes the target for every possible initial state simultaneously.

\subsection{Non-identifiability Under a Uniform Parameter-Derivative Bound}

The preceding witness can be chosen to satisfy an arbitrarily prescribed uniform bound on the parameter derivative of the Hamiltonian.

\begin{theorem}
For every pure state $\rho_0$ and every $L>0$, there exist
$H_0,H_1\in\mathcal D_T^\infty$ such that
\begin{align}
\sup_{(\theta,t)\in I\times[0,T]}
\|\partial_\theta H_j(\theta,t)\|_{\mathrm{op}}
\le L,
\qquad j=0,1,
\label{eq:derivative_bound}
\end{align}
while
\begin{align}
E_{\mathrm{ph}}[H_0]&=E_{\mathrm{ph}}[H_1],
&
G_{\rho_0}[H_0]&\ne G_{\rho_0}[H_1].
\end{align}
No corresponding uniform bound on $\|H\|_{\mathrm{op}}$ is assumed.
\end{theorem}

\begin{proof}
Take $H_0=0$ and
\begin{align}
H_1(\theta,t)
=
\alpha\theta\,\dot f(t)A,
\end{align}
where $f(t)=\sin(2\pi t/T)$ and $A$ is the Hermitian two-level operator introduced above. Since
$\|A\|_{\mathrm{op}}=1$,
\begin{align}
\|\partial_\theta H_1(\theta,t)\|_{\mathrm{op}}
=
|\alpha|\,|\dot f(t)|
\le
|\alpha|\,\|\dot f\|_{\infty}.
\end{align}
Choosing a nonzero $\alpha$ with
$|\alpha|\le L/\|\dot f\|_{\infty}$ gives Eq.~\eqref{eq:derivative_bound}. The periodic primitive $f$ has zero net change over each period, so the complete starting-time-indexed endpoint propagators of $H_1$ are identical to those of $H_0$. On the other hand,
\begin{align}
G_{\rho_0}[H_1](\theta)
=
\frac{\alpha^2}{2}
\operatorname{Var}_{\rho_0}(A),
\end{align}
which is nonzero for $\alpha\ne0$.
\end{proof}

Thus, a uniform bound on the parameter derivative restricts the magnitude of the target but does not remove the fibre ambiguity responsible for its non-identifiability.

\subsection{Pointwise Full-Range Ambiguity in the Identity Fibre}

Let $\mathbf 1_{\mathrm{ph}}$ denote the constant identity observation,
\begin{align}
\mathbf 1_{\mathrm{ph}}(\theta,\tau)=I_d,
\end{align}
and define the corresponding fibre
\begin{align}
\mathcal E_{\mathbf 1}
=
\left\{
H\in\mathcal D_T^\infty:
E_{\mathrm{ph}}[H]\equiv\mathbf 1_{\mathrm{ph}}
\right\}.
\end{align}

\begin{theorem}
For every $\theta_0\in I$,
\begin{align}
\left\{
G_{\rho_0}[H](\theta_0):
H\in\mathcal E_{\mathbf 1}
\right\}
=
[0,\infty).
\label{eq:full_range}
\end{align}
\end{theorem}

\begin{proof}
For $a\ge0$, define
\begin{align}
q_a(\theta)=a(\theta-\theta_0),
\qquad
U_{\theta,a}(t)
=
e^{-iq_a(\theta)f(t)A},
\end{align}
with $f(t)=\sin(2\pi t/T)$. The associated Hamiltonian is
\begin{align}
H_{\theta,a}(t)
=
q_a(\theta)\dot f(t)A.
\end{align}
Since $f$ is $T$-periodic,
\begin{align}
U_{\theta,a}(\tau+T,\tau)
=
I_d
\end{align}
for every $\tau$, and hence
$E_{\mathrm{ph}}[H_a]\equiv\mathbf 1_{\mathrm{ph}}$. At $\theta=\theta_0$,
\begin{align}
q_a'(\theta_0)=a,
\end{align}
so that
\begin{align}
G_{\rho_0}[H_a](\theta_0)
=
\frac{a^2}{2}
\operatorname{Var}_{\rho_0}(A).
\end{align}
For the chosen witness,
$\operatorname{Var}_{\rho_0}(A)=1$. Therefore every value in
$[0,\infty)$ is obtained by choosing
$a=\sqrt{2r}/T$ for a prescribed $r\ge0$.
\end{proof}

Equation~\eqref{eq:full_range} is a pointwise statement. It does not imply that an arbitrary prescribed function of $\theta$ can be realized throughout $I$ while remaining in the identity fibre.

\subsection{Small Hamiltonian Amplitude Does Not Control Parameter-Space Geometry}

The absence of identifiability is distinct from the behavior of the forward map under small perturbations of the Hamiltonian itself.

\begin{proposition}
Let $I_0\Subset I$ be compact and let
$\theta_0\in\operatorname{int}I_0$. There exists a sequence
$H_N\in\mathcal D_T^\infty$ such that
\begin{align}
\|H_N\|_{C^0(I_0\times[0,T])}&\longrightarrow0,
&
E_{\mathrm{ph}}[H_N]&\equiv\mathbf 1_{\mathrm{ph}},
\label{eq:small_H}
\\
G_{\rho_0}[H_N](\theta_0)&\longrightarrow\infty .
\label{eq:large_G}
\end{align}
\end{proposition}

\begin{proof}
Let $0<\gamma<1$ and define
\begin{align}
q_N(\theta)
=
N^{-\gamma}
\sin\!\left[N(\theta-\theta_0)\right],
\qquad
H_N(\theta,t)
=
q_N(\theta)\dot f(t)A .
\end{align}
The corresponding propagator is
\begin{align}
U_{\theta,N}(t)
=
e^{-iq_N(\theta)f(t)A},
\end{align}
so the endpoint observation is identically $\mathbf 1_{\mathrm{ph}}$. Moreover,
\begin{align}
q_N'(\theta_0)=N^{1-\gamma},
\end{align}
and therefore
\begin{align}
G_{\rho_0}[H_N](\theta_0)
=
\frac{1}{2}
N^{2(1-\gamma)}
\operatorname{Var}_{\rho_0}(A)
\longrightarrow\infty .
\end{align}
At the same time,
\begin{align}
\|H_N\|_{C^0(I_0\times[0,T])}
\le
N^{-\gamma}\|\dot f\|_{\infty}
\|A\|_{\mathrm{op}}
\longrightarrow0.
\end{align}
\end{proof}

Thus small Hamiltonian amplitude alone places no useful upper bound on the parameter-space metric. The mechanism is the independent growth of the parameter derivative, rather than any large instantaneous Hamiltonian amplitude.

\subsection{A Priori Bound from a Parameter-Derivative Constraint}

Although a derivative bound does not restore identifiability, it does provide a uniform bound on the target functional.

\begin{proposition}
Suppose that
\begin{align}
\|\partial_\theta H_\theta(t)\|_{\mathrm{op}}
\le L
\end{align}
for all $(\theta,t)\in I\times[0,T]$. Then
\begin{align}
G_{\rho_0}[H](\theta)
\le
\frac{T^2L^2}{3}.
\label{eq:G_apriori}
\end{align}
\end{proposition}

\begin{proof}
Differentiating the propagator with respect to $\theta$ gives
\begin{align}
K_\theta(t)
=
\int_0^t
U_\theta^\dagger(s)
\partial_\theta H_\theta(s)
U_\theta(s)\,ds .
\end{align}
Unitary invariance of the operator norm therefore yields
\begin{align}
\|K_\theta(t)\|_{\mathrm{op}}
\le tL .
\end{align}
For any Hermitian operator $K$,
\begin{align}
\operatorname{Var}_{\rho_0}(K)
\le
\|K\|_{\mathrm{op}}^2 .
\end{align}
Consequently,
\begin{align}
G_{\rho_0}[H](\theta)
\le
\frac{1}{T}
\int_0^T t^2L^2\,dt
=
\frac{T^2L^2}{3}.
\end{align}
\end{proof}

The bound in Eq.~\eqref{eq:G_apriori} and identifiability are logically independent. A restricted Hamiltonian class may render the target uniformly bounded without making it recoverable from the endpoint observation.

\subsection{Endpoint Jets Do Not Restore Identifiability}

\begin{corollary}
Knowledge of the complete $C^\infty$ endpoint-data map, including all of its derivatives with respect to $\theta$ and $\tau$, does not restore identifiability of $G_{\rho_0}$.
\end{corollary}

\begin{proof}
The analytic witness constructed above satisfies
\begin{align}
E_{\mathrm{ph}}[H_\alpha]
=
E_{\mathrm{ph}}[H_0]
\end{align}
as identical $C^\infty$ maps on $I\times\mathbb T_T$. Hence their derivatives of every order coincide at every point, whereas
$G_{\rho_0}[H_\alpha]\ne G_{\rho_0}[H_0]$ for generic $\alpha$. Therefore no deterministic reconstruction rule based solely on the complete endpoint-data map, or equivalently on its full $C^\infty$ jet, can determine $G_{\rho_0}$ on the stated Hamiltonian class.
\end{proof}

\section{Deterministic Observation-Factorization Hierarchy}

It is useful to compare the three observation schemes in terms of the information they retain about the underlying dynamics. We introduce the deterministic factorization preorder
\begin{align}
A\preceq B
\quad\Longleftrightarrow\quad
A=R\circ B
\end{align}
for some deterministic post-processing map $R$. Thus, $B$ is at least as informative as $A$ if the latter can be obtained from the former without access to any additional dynamical information. We write $A\prec B$ when $A\preceq B$ but $B\not\preceq A$. This is purely a set-theoretic notion of observational sufficiency; no continuity, smoothness, or other regularity of $R$ is assumed.

The observation maps introduced above obey the factorization relations
\begin{align}
E_0&=R_0\circ E_{\mathrm{ph}},
&
[R_0(\Gamma)](\theta)&=\Gamma(\theta,0)=M_\theta ,
\label{eq:factor_E0}
\\
E_{\mathrm{ph}}&=R_{\mathrm{ph}}\circ E_{\mathrm{full}},
&
[R_{\mathrm{ph}}(U)](\theta,\tau)
&=U_\theta(\tau)U_\theta(T)U_\theta^\dagger(\tau).
\label{eq:factor_Eph}
\end{align}
Consequently, the starting-time-indexed endpoint data contain the fixed-origin monodromy as a special case, while the full intra-period propagator contains the endpoint data.

These inclusions are strict. First, $E_{\mathrm{ph}}$ contains information that is absent from the fixed-origin monodromy. To see this, let
\begin{align}
M=e^{-icT\sigma_z},
\quad
W(t)=e^{-i\varepsilon f(t)\sigma_x},
\quad
f(t)=\sin(2\pi t/T),
\end{align}
and define
\begin{align}
U_0(t)&=e^{-ict\sigma_z},
&
U_1(t)&=W(t)U_0(t).
\end{align}
Since $W(0)=W(T)=I$, both propagators have the same monodromy,
\begin{align}
U_0(T)=U_1(T)=M.
\end{align}
Nevertheless, for generic $\varepsilon$,
\begin{align}
U_1(\tau)MU_1^\dagger(\tau)
=
W(\tau)MW^\dagger(\tau)
\neq M,
\end{align}
so that their starting-time-indexed endpoint observations differ. Hence
\begin{align}
E_0\prec E_{\mathrm{ph}}.
\end{align}

The second inclusion is also strict. The analytic commuting family constructed above satisfies
\begin{align}
E_{\mathrm{ph}}[H_\alpha]=E_{\mathrm{ph}}[H_0]
\end{align}
for all values of the deformation parameter $\alpha$, whereas the corresponding intra-period propagators are distinct. Therefore,
\begin{align}
E_{\mathrm{ph}}\prec E_{\mathrm{full}}.
\end{align}

We thus obtain the strict hierarchy
\begin{align}
\boxed{
E_0\prec E_{\mathrm{ph}}\prec E_{\mathrm{full}}.
}
\label{eq:observation_hierarchy}
\end{align}
The hierarchy separates three physically distinct levels of dynamical information: the one-period propagator from a fixed reference time, its complete starting-time-indexed conjugation orbit, and the full intra-period evolution.

The period-averaged Fubini--Study pullback component considered here does not factor through either of the first two observation levels,
\begin{align}
G_{\rho_0}\not\preceq E_0,
\qquad
G_{\rho_0}\not\preceq E_{\mathrm{ph}},
\end{align}
whereas it does factor through the full propagator:
\begin{align}
G_{\rho_0}
=
R_{\rho_0}\circ E_{\mathrm{full}},
\end{align}
with
\begin{align}
[R_{\rho_0}(U)](\theta)
=
\frac{1}{T}
\int_0^T
\operatorname{Var}_{\rho_0}
\!\left(
iU_\theta^\dagger(t)\partial_\theta U_\theta(t)
\right)dt .
\label{eq:G_factor_full}
\end{align}
Thus, complete knowledge of the intra-period propagator is sufficient to determine the target functional, whereas knowledge of all one-period endpoint propagators, even for every starting time, is not.

Equation \eqref{eq:observation_hierarchy} therefore identifies the precise information threshold relevant to the inverse problem. The failure of identifiability is not caused simply by insufficient knowledge of the Floquet operator at one reference time; it persists even after the entire starting-time-indexed family of one-period endpoint propagators is supplied. The additional information required is genuinely intra-period dynamical information. We do not claim that $E_{\mathrm{full}}$ is a minimal sufficient observation scheme.

\section{Physical Interpretation and Discussion}

The results above distinguish several logically different sources of ambiguity in Floquet inverse problems. It is important to separate these mechanisms because the non-identifiability established here does not rely on the usual ambiguity associated with choosing a Floquet Hamiltonian.

\paragraph{Distinct sources of Floquet ambiguity.}
Three levels of ambiguity may arise. First, the \emph{Floquet-logarithm ambiguity} reflects the fact that a one-period unitary $M$ does not, in general, determine a unique time-independent generator through $M=e^{-iTH}$. Second, even when a particular monodromy $M$ is fixed, there are generally many $T$-periodic propagator histories having that same monodromy. We refer to this as the \emph{monodromy-compatible propagator ambiguity}. Third, and most importantly for the present work, knowledge of the entire starting-time-indexed family
\begin{align}
\Gamma(t)=U(t)MU^\dagger(t)
\end{align}
still does not determine the unitary representative $U(t)$. The remaining freedom is the \emph{centralizer-fiber ambiguity} characterized above.

The theorem is fundamentally a statement about this third level. If $U$ is replaced by $UV$, with $V(t)$ taking values in the centralizer of $M$, the endpoint observation is unchanged. Nevertheless, the body-frame generator transforms as
\begin{align}
K\longmapsto
V^\dagger K V+iV^\dagger\partial_\theta V .
\end{align}
Both the conjugation of $K$ and the inhomogeneous parameter-derivative term can therefore modify the period-averaged Fubini--Study metric while leaving the complete endpoint observation unchanged. The obstruction consequently survives even when the Floquet logarithm is unambiguous and when the dynamics are Abelian in time.

\paragraph{Scope of the non-identifiability result.}
It is useful to distinguish three statements. \emph{Exact non-identifiability} means that two Hamiltonian families can produce exactly the same observation $E_{\mathrm{ph}}$ while yielding different values of $G_{\rho_0}$. \emph{Unbounded fiber ambiguity} is the stronger statement that a single endpoint-data fiber can contain a family for which the target varies without bound. \emph{Generic non-identifiability} would assert that such a failure occurs for a suitably large, for example open or dense, subset of Hamiltonians.

The present work establishes the first two statements. In particular, the analytic two-level construction provides an explicit fiber along which $G_{\rho_0}$ is unbounded. No claim of generic non-identifiability is required for the theorem, and none is made here. Likewise, the result is a deterministic identifiability statement; it does not by itself imply a quantitative instability theorem for the reconstruction problem in the presence of noisy or incomplete endpoint data.

\paragraph{Restricted dynamical classes.}
The negative result concerns the unrestricted smooth periodic Hamiltonian class. Additional structural assumptions can remove the ambiguity. For example, for a time-independent Hamiltonian whose spectrum remains in a fixed arc of the unit circle on which a smooth logarithm branch is available, the monodromy can determine the generator through that branch. If, in addition, the parameter dependence is commuting, so that
\begin{align}
[H_\theta,\partial_\theta H_\theta]=0,
\end{align}
the parameter response is fixed by the generator and the corresponding metric reduces to the familiar static expression. Thus, the non-identifiability theorem should not be interpreted as a statement about all restricted Floquet models; rather, it identifies an obstruction that is unavoidable on the unrestricted dynamical class considered here.

\paragraph{A fortiori for coarser observations.}
The obstruction also persists for any observation obtained by further processing the endpoint data. Let $Q$ be a deterministic map from the unitary endpoint data to a coarser representation, such as projective-unitary or quantum-channel data, so that
\begin{align}
E_{\mathrm{quotient}}=Q\circ E_{\mathrm{ph}}.
\end{align}
If $G_{\rho_0}$ could be reconstructed deterministically from $E_{\mathrm{quotient}}$, then it would also factor through $E_{\mathrm{ph}}$, since
\begin{align}
G_{\rho_0}
=
R\circ E_{\mathrm{quotient}}
=
(R\circ Q)\circ E_{\mathrm{ph}}.
\end{align}
This contradicts the non-identifiability theorem. Hence passing from the full endpoint unitary to a quotient or channel representation cannot restore information that is already absent from $E_{\mathrm{ph}}$.

\paragraph{Infinitesimal viewpoint.}
The global fiber structure also has a direct dynamical interpretation. Differentiating the observed orbit gives
\begin{align}
\dot{\Gamma}(t)=-i[H(t),\Gamma(t)].
\end{align}
Thus the endpoint data constrain only the component of the instantaneous generator that acts nontrivially through the adjoint action on $\Gamma(t)$. Components belonging to the kernel of
$\operatorname{ad}_{\Gamma(t)}$ are invisible to this equation. This kernel is precisely the Lie-algebra centralizer
$\mathfrak z(\Gamma(t))$. The identity
\begin{align}
\mathfrak z(\Gamma(t))
=
U(t)\mathfrak z(M)U^\dagger(t)
\end{align}
connects this local description with the global centralizer-fiber characterization. The infinitesimal relation therefore provides a useful physical interpretation of the fiber, although the global theorem does not rely on the infinitesimal argument.

\paragraph{The obstruction is not a scalar phase.}
Finally, the ambiguity identified here should not be confused with the ordinary $U(1)$ phase freedom. A scalar contribution to the body-frame generator changes $K$ by a multiple of the identity, which leaves its variance, and hence the Fubini--Study metric, unchanged. The relevant fiber contains, in general, noncentral directions of the unitary group and therefore represents a physically distinct ambiguity.

This distinction is explicit in the two-level witness. There, the deformation is generated by the nontrivial operator
\begin{align}
A=
|\psi_0\rangle\langle\phi|
+
|\phi\rangle\langle\psi_0|,
\end{align}
for which
$\operatorname{Var}_{\rho_0}(A)=1$. The corresponding deformation therefore changes the physical ray
$[U_\theta(t)|\psi_0\rangle]$ during the driving cycle while leaving every one-period starting-time endpoint propagator unchanged. The non-identifiability is consequently a genuine loss of intra-period dynamical information, rather than a removable projective gauge freedom.

\section{CONCLUSION}

We have established a global, worst-case non-identifiability theorem for periodically driven finite-dimensional quantum systems. On the unrestricted smooth periodic Hamiltonian class, exact one-period propagators indexed by every starting time and parameter value do not, in general, determine the period-averaged Fubini--Study metric component of a fixed-initial-state trajectory.

The key structural result is an exact characterization of the observational fibres. The starting-time-indexed endpoint data determine the conjugation-orbit path of the monodromy but not its particular unitary representative. On each fixed-monodromy slice, the fibres are precisely the right orbits generated by smooth parameter-dependent based loops valued in the pointwise centralizer of the monodromy. The Fubini--Study functional is not invariant under this fibre action and therefore does not factor through the endpoint observation map.

An explicit real-analytic two-level witness demonstrates the obstruction in an Abelian setting: all Hamiltonian time slices commute, the monodromy depends nontrivially on the parameter, and yet a continuous family of distinct Hamiltonians produces exactly identical starting-time-indexed endpoint data while yielding different period-averaged Fubini--Study geometry. The difference can grow quadratically along the observational fibre. Non-identifiability also persists under a uniform bound on the parameter derivative of the Hamiltonian.

The result is a global worst-case statement, not a claim of generic non-identifiability or experimental impossibility. It identifies a precise information gap between complete starting-time-indexed one-period endpoint data and the full intra-period propagator. The missing information is the unitary lift of the observed monodromy orbit, rather than ordinary scalar phase freedom.

Thus, complete one-period endpoint observations can determine substantial Floquet micromotion information without determining all micromotion-sensitive quantum geometry. Additional structural assumptions or genuinely intra-period measurements are required to restore identifiability of the target considered here.

\bibliography{references}

@article{Shirley1965,
  author    = {J. H. Shirley},
  title     = {Solution of the Schr\"odinger Equation with a Hamiltonian Periodic in Time},
  journal   = {Phys. Rev.},
  volume    = {138},
  pages     = {B979--B987},
  year      = {1965},
  doi       = {10.1103/PhysRev.138.B979},
}

@article{Sambe1973,
  author    = {H. Sambe},
  title     = {Steady States and Quasienergies of a Quantum-Mechanical System in an Oscillating Field},
  journal   = {Phys. Rev. A},
  volume    = {7},
  pages     = {2203--2215},
  year      = {1973},
  doi       = {10.1103/PhysRevA.7.2203},
}

@article{GoldmanDalibard2014,
  author    = {N. Goldman and J. Dalibard},
  title     = {Periodically Driven Quantum Systems: Effective Hamiltonians and Engineered Gauge Fields},
  journal   = {Phys. Rev. X},
  volume    = {4},
  pages     = {031027},
  year      = {2014},
  doi       = {10.1103/PhysRevX.4.031027},
}

@article{Rudner2013,
  author    = {M. S. Rudner and N. H. Lindner and E. Berg and M. Levin},
  title     = {Anomalous Edge States and the Bulk-Edge Correspondence for Periodically Driven Two-Dimensional Systems},
  journal   = {Phys. Rev. X},
  volume    = {3},
  pages     = {031005},
  year      = {2013},
  doi       = {10.1103/PhysRevX.3.031005},
}

@article{ProvostVallee1980,
  author    = {J. P. Provost and G. Vall\'ee},
  title     = {Riemannian Structure on Manifolds of Quantum States},
  journal   = {Commun. Math. Phys.},
  volume    = {76},
  pages     = {289--301},
  year      = {1980},
  doi       = {10.1007/BF02193559},
}

@article{Kolodrubetz2017,
  author    = {M. Kolodrubetz and D. Sels and P. Mehta and A. Polkovnikov},
  title     = {Geometry and Non-Adiabatic Response in Quantum and Classical Systems},
  journal   = {Phys. Rep.},
  volume    = {697},
  pages     = {1--87},
  year      = {2017},
  doi       = {10.1016/j.physrep.2017.07.001},
}

@article{Berry1984,
  author    = {M. V. Berry},
  title     = {Quantal Phase Factors Accompanying Adiabatic Changes},
  journal   = {Proc. R. Soc. Lond. A},
  volume    = {392},
  pages     = {45--57},
  year      = {1984},
  doi       = {10.1098/rspa.1984.0023},
}

@article{AharonovAnandan1987,
  author    = {Y. Aharonov and J. Anandan},
  title     = {Phase Change during a Cyclic Quantum Evolution},
  journal   = {Phys. Rev. Lett.},
  volume    = {58},
  pages     = {1593--1597},
  year      = {1987},
  doi       = {10.1103/PhysRevLett.58.1593},
}

@article{ZanardiPaunkovic2006,
  author    = {P. Zanardi and N. Paunkovi\'c},
  title     = {Ground State Overlap and Quantum Phase Transitions},
  journal   = {Phys. Rev. E},
  volume    = {74},
  pages     = {031123},
  year      = {2006},
  doi       = {10.1103/PhysRevE.74.031123},
}

@article{ZanardiGiordaCozzini2007,
  author    = {P. Zanardi and P. Giorda and M. Cozzini},
  title     = {Information-Theoretic Differential Geometry of Quantum Phase Transitions},
  journal   = {Phys. Rev. Lett.},
  volume    = {99},
  pages     = {100603},
  year      = {2007},
  doi       = {10.1103/PhysRevLett.99.100603},
}

@article{Floquet1883,
  author  = {Floquet, Gaston},
  title   = {Sur les {\'e}quations diff{\'e}rentielles lin{\'e}aires {\`a} coefficients p{\'e}riodiques},
  journal = {Annales Scientifiques de l'{\'E}cole Normale Sup{\'e}rieure},
  volume  = {12},
  pages   = {47--88},
  year    = {1883}
}

@article{BukovDAlessioPolkovnikov2015,
  author  = {Bukov, Marin and D'Alessio, Luca and Polkovnikov, Anatoli},
  title   = {Universal High-Frequency Behavior of Periodically Driven Systems: from Dynamical Stabilization to Floquet Engineering},
  journal = {Advances in Physics},
  volume  = {64},
  number  = {2},
  pages   = {139--226},
  year    = {2015},
  doi     = {10.1080/00018732.2015.1055918}
}

@article{Eckardt2017,
  author  = {Eckardt, Andr{\'e}},
  title   = {Colloquium: Atomic Quantum Technologies in Periodically Driven Optical Lattices},
  journal = {Reviews of Modern Physics},
  volume  = {89},
  number  = {1},
  pages   = {011004},
  year    = {2017},
  doi     = {10.1103/RevModPhys.89.011004}
}

@article{BraunsteinCaves1994,
  author  = {Braunstein, Samuel L. and Caves, Carlton M.},
  title   = {Statistical Distance and the Geometry of Quantum States},
  journal = {Physical Review Letters},
  volume  = {72},
  number  = {22},
  pages   = {3439--3443},
  year    = {1994},
  doi     = {10.1103/PhysRevLett.72.3439}
}

@article{Holthaus2016,
  author  = {Holthaus, M.},
  title   = {Floquet engineering with quasienergy bands of periodically driven optical lattices},
  journal = {J. Phys. B: At. Mol. Opt. Phys.},
  volume  = {49},
  pages   = {013001},
  year    = {2016},
  doi     = {10.1088/0953-4075/49/1/013001}
}

@article{OkaKitamura2019,
  author  = {Oka, T. and Kitamura, S.},
  title   = {Floquet engineering of quantum materials},
  journal = {Annu. Rev. Condens. Matter Phys.},
  volume  = {10},
  pages   = {387--408},
  year    = {2019},
  doi     = {10.1146/annurev-conmatphys-031218-013423}
}

@article{NathanRudner2015,
  author  = {Nathan, F. and Rudner, M. S.},
  title   = {Topological singularities and the general classification of Floquet--Bloch systems},
  journal = {New J. Phys.},
  volume  = {17},
  pages   = {125014},
  year    = {2015},
  doi     = {10.1088/1367-2630/17/12/125014}
}

@article{Simon1983,
  author  = {Simon, B.},
  title   = {Holonomy, the quantum adiabatic theorem, and Berry's phase},
  journal = {Phys. Rev. Lett.},
  volume  = {51},
  pages   = {2167--2170},
  year    = {1983},
  doi     = {10.1103/PhysRevLett.51.2167}
}

@article{Kitagawa2010,
  author  = {Kitagawa, T. and Berg, E. and Rudner, M. S. and Demler, E.},
  title   = {Topological characterization of periodically driven quantum systems},
  journal = {Phys. Rev. B},
  volume  = {82},
  pages   = {235114},
  year    = {2010},
  doi     = {10.1103/PhysRevB.82.235114}
}

@article{LindnerRefaelGalitski2011,
  author  = {Lindner, N. H. and Refael, G. and Galitski, V.},
  title   = {Floquet topological insulator in semiconductor quantum wells},
  journal = {Nat. Phys.},
  volume  = {7},
  pages   = {490--495},
  year    = {2011},
  doi     = {10.1038/nphys1926}
}

@article{EckardtAnisimovas2015,
  author  = {Eckardt, A. and Anisimovas, E.},
  title   = {High-frequency approximation for periodically driven quantum systems from a Floquet-space perspective},
  journal = {New J. Phys.},
  volume  = {17},
  pages   = {093039},
  year    = {2015},
  doi     = {10.1088/1367-2630/17/9/093039}
}

@article{Bleu2018,
  author  = {Bleu, O. and Solnyshkov, D. D. and Malpuech, G.},
  title   = {Measuring the quantum geometric tensor in two-dimensional photonic and exciton-polariton systems},
  journal = {Phys. Rev. B},
  volume  = {97},
  pages   = {195422},
  year    = {2018},
  doi     = {10.1103/PhysRevB.97.195422}
}

@article{Tan2019,
  author  = {Tan, X. and Zhang, D.-W. and Yang, Z. and Chu, J. and Zhu, Y.-Q. and Li, D. and Yang, X. and Song, S. and Han, Z. and Li, Z. and Dong, Y. and Yu, H.-F. and Yan, H. and Zhu, S.-L. and Yu, Y.},
  title   = {Experimental Measurement of the Quantum Metric Tensor and Related Topological Phase Transition with a Superconducting Qubit},
  journal = {Phys. Rev. Lett.},
  volume  = {122},
  pages   = {210401},
  year    = {2019},
  doi     = {10.1103/PhysRevLett.122.210401}
}

@article{Yu2019,
  author  = {Yu, M. and Yang, P. and Gong, M. and Cao, Q. and Lu, Q. and Liu, H. and Zhang, S. and Plenio, M. B. and Jelezko, F. and Ozawa, T. and Goldman, N. and Cai, J.-M.},
  title   = {Experimental measurement of the quantum geometric tensor using coupled qubits in diamond},
  journal = {Natl. Sci. Rev.},
  volume  = {7},
  pages   = {254--260},
  year    = {2020},
  doi     = {10.1093/nsr/nwz193}
}

@article{Gianfrate2020,
  author  = {Gianfrate, A. and Bleu, O. and Dominici, L. and Ardizzone, M. and De Giorgi, M. and Ballarini, D. and Lerario, G. and West, K. W. and Pfeiffer, L. N. and Solnyshkov, D. D. and Sanvitto, D. and Malpuech, G.},
  title   = {Measurement of the quantum geometric tensor and of the anomalous Hall drift},
  journal = {Nature},
  volume  = {578},
  pages   = {381--385},
  year    = {2020},
  doi     = {10.1038/s41586-020-1989-2}
}
\bibliographystyle{unsrt}

\end{document}